\documentclass[sn-mathphys-num]{sn-jnl}

\usepackage{graphicx}%
\usepackage{multirow}%
\usepackage{amsmath,amssymb,amsfonts}%
\usepackage{amsthm}%
\usepackage{mathrsfs}%
\usepackage[title]{appendix}%
\usepackage{xcolor}%
\usepackage{textcomp}%
\usepackage{manyfoot}%
\usepackage{booktabs}%
\usepackage{algorithm}%
\usepackage{algorithmicx}%
\usepackage{algpseudocode}%
\usepackage{listings}%
\usepackage{bigstrut,bigdelim,multirow,diagbox}%

\theoremstyle{thmstyleone}%
\newtheorem{theorem}{Theorem}[section]
\newtheorem{proposition}[theorem]{Proposition}%
\newtheorem{lemma}[theorem]{Lemma}

\theoremstyle{thmstyletwo}%

\theoremstyle{thmstylethree}%
\newtheorem{definition}{Definition}%

\begin{document}

\title[Article Title]{An Attack on $p$-adic Lattice Public-key Encryption Cryptosystems and Signature Schemes}


\author*[1,2]{\fnm{Chi} \sur{Zhang}}\email{zhangchi171@mails.ucas.ac.cn}

%

\affil*[1]{\orgdiv{Key Laboratory of Mathematics Mechanization, NCMIS, Academy of Mathematics and Systems Science}, \orgname{Chinese Academy of Sciences}, \orgaddress{\city{Beijing}, \postcode{100190}, \country{People's Republic of China}}}

\affil[2]{\orgdiv{School of Mathematical Sciences}, \orgname{University of Chinese Academy of Sciences}, \orgaddress{\city{Beijing}, \postcode{100049}, \country{Peoples Republic of China}}}



\abstract{Lattices have many significant applications in cryptography. In 2021, the $p$-adic signature scheme and public-key encryption cryptosystem were introduced. They are based on the Longest Vector Problem (LVP) and the Closest Vector Problem (CVP) in $p$-adic lattices. These problems are considered to be challenging and there are no known deterministic polynomial time algorithms to solve them. In this paper, we improve the LVP algorithm in local fields. The modified LVP algorithm is a deterministic polynomial time algorithm when the field is totally ramified and $p$ is a polynomial in the rank of the input lattice. We utilize this algorithm to attack the above schemes so that we are able to forge a valid signature of any message and decrypt any ciphertext. Although these schemes are broken, this work does not mean that $p$-adic lattices are not suitable in constructing cryptographic primitives. We propose some possible modifications to avoid our attack at the end of this paper.}




\maketitle

\section{Introduction}\label{se-1}

Lattices were investigated since the late $18$th century. They are important research objects in geometric number theory and have many applications. There are two famous computational problems in Euclidean lattices: the Shortest Vector Problem (SVP) and the Closest Vector Problem (CVP). Much research has been done to study these problems. Van Emde Boas \cite{ref-6} proved that the SVP is NP-hard in the $l_{\infty}$ norm and the CVP is NP-hard in the $l_2$ norm by reducing them to the weak partition problem. Ajtai \cite{ref-7} proved that the SVP is NP-hard under randomized reductions in the $l_2$ norm. These hard problems are useful in constructing cryptographic primitives. The GGH scheme \cite{ref-8}, the NTRU scheme \cite{ref-9} and the LWE scheme \cite{ref-10} are based on these problems.\par
Since Peter Shor \cite{ref-11} proved that the classical public-key cryptosystems such as RSA and ElGamal would be broken under future quantum computer, researchers have been dedicated to finding cryptographic primitives which are quantum-resistant. In 2022, NIST \cite{ref-16} announced four algorithms which passed the third round of post-quantum cryptography standardization solicitation. They are CRYSTALS-Kyber \cite{ref-12}, CRYSTALS-Dilithium \cite{ref-13}, Falcon \cite{ref-14} and SPHINCS$^{+}$ \cite{ref-15}. Three of them are lattice-based and one of them is hash-based. The lack of diversity among post-quantum assumptions is widely recognized as a big, open issue in the field. Therefore, finding new  post-quantum assumptions is of vital significance. Just as Neal Koblitz said in the Preface of his book \cite{ref-5.5}: ``But in the rapidly growing field of cryptography it is worthwhile to continually explore new one-way constructions coming from different areas of mathematics.''\par
Lattices can also be defined in local fields such as $p$-adic fields, see \cite{ref-5}. Interestingly, $p$-adic lattices possess some properties which lattices in Euclidean spaces do not have. For example, $\mathbb{Q}_p^n$ can be viewed as a field for every integer $n\ge1$, while the famous Frobenius Theorem asserts that $\mathbb{R}^n$ can be viewed as a field only when $n=1,2,4$. The case $n=2$ is the field of complex numbers and when $n=4$, the field is non-commutative (i.e., Hamilton quaternions). Besides, every $p$-adic lattice has an orthogonal basis, which is proved by Zhang et al. \cite{ref-3.5}, whereas lattices in Euclidean spaces lack such bases in general.\par
However, $p$-adic lattices do not gain as much attention as lattices in Euclidean spaces. It is not until recently that the applications of $p$-adic lattices in cryptography get developed. In 2018, motivated by the Shortest Vector Problem and the Closest Vector Problem of lattices in Euclidean spaces, Deng et al. \cite{ref-1',ref-1} introduced two new computational problems of $p$-adic lattices in local fields, the Longest Vector Problem (LVP) and the Closest Vector Problem (CVP). They proposed deterministic exponential time algorithms to solve them. The LVP algorithm takes $O(p^m)$ many $p$-adic absolute value computations, where $m$ is the rank of the input lattice. These new problems are considered to be challenging and potentially applicable for constructing public-key encryption cryptosystems and signature schemes. Moreover, as the $p$-adic analogues of the lattices in Euclidean spaces, it is reasonable to expect these problems to be quantum-resistant, which  might provide new alternative candidates to construct post-quantum cryptographic primitives.\par
Although there are no known deterministic polynomial time algorithms to solve the LVP or the CVP, Deng et al. \cite{ref-2} proved that the LVP and the CVP can be solved efficiently with  the help of an orthogonal basis. Hence orthogonal bases can be used to construct trapdoor information for cryptographic schemes. In 2021, by introducing a trapdoor function with an orthogonal basis of a $p$-adic lattice, Deng et al. \cite{ref-2} constructed the first signature scheme and public-key encryption cryptosystem based on $p$-adic lattices.\par
The experimental results demonstrated that the new schemes achieve good efficiency. As for security, there is no formal proof for these schemes since the problems in local fields are relatively very new. Instead, Deng et al. \cite{ref-2} discussed some possible attacks on their schemes, including recovering a uniformizer, finding an orthogonal basis, solving CVP-instances and modulo $p$ attack. Therefore, the security of their schemes remains an open problem.\par
This paper is organized as follows. In Section \ref{se-2}, we recall some basic definitions and briefly review the $p$-adic signature scheme and public-key encryption cryptosystem. In section \ref{se-3}, we present an orthogonalization algorithm in totally ramified fields. In section \ref{se-4.1} and \ref{se-4.2}, we use this algorithm to attack the $p$-adic signature scheme and public-key encryption cryptosystem. In section \ref{se-4.3}, we demonstrate a toy example of this attack, and list the time consumption of finding a uniformizer for different public polynomials.

\section{Preliminaries}\label{se-2}

\subsection{Basic Facts about Local Fields}
In this subsection, we recall some basic facts about local fields, see \cite{ref-1',ref-1}. More details about local fields can be found in \cite{ref-4,ref-4.5}.\par
Let $p$ be a prime number. For $x\in\mathbb{Q}$ with $x\ne0$, write $x=p^{t}\cdot\frac{a}{b}$ with $t,a,b\in\mathbb{Z}$ and $p\nmid ab$. Define $\left|x\right|_p=p^{-t}$ and $\left|0\right|_p=0$. Then $|\cdot|_p$ is a non-Archimedean absolute value on $\mathbb{Q}$. Namely, we have:

\begin{enumerate}
\item $\left|x\right|_p\ge0$ and $\left|x\right|_p=0$ if and only if $x=0$;
\item $\left|xy\right|_p=\left|x\right|_p\left|y\right|_p$;
\item $\left|x+y\right|_p\le\max{\left\{\left|x\right|_p,\left|y\right|_p\right\}}$. If $\left|x\right|_p\ne\left|y\right|_p$, then $\left|x+y\right|_p=\max{\left\{\left|x\right|_p,\left|y\right|_p\right\}}$.
\end{enumerate}

Let $\mathbb{Q}_p$ be the completion of $\mathbb{Q}$ with respect to $|\cdot|_p$. Denote
$$\mathbb{Z}_p=\left\{x\in\mathbb{Q}_p:\left|x\right|_p\le1\right\}.$$
We have
$$\mathbb{Q}_p=\left\{\sum_{i=j}^{\infty}a_ip^i:a_i\in\left\{0,1,2,\dots,p-1\right\},i\ge j,j\in\mathbb{Z}\right\},$$
and
$$\mathbb{Z}_p=\left\{\sum_{i=0}^{\infty}a_ip^i:a_i\in\left\{0,1,2,\dots,p-1\right\},i\ge 0\right\}.$$
$\mathbb{Z}_p$ is compact and $\mathbb{Q}_p$ is locally compact. $\mathbb{Z}_p$ is a discrete valuation ring, it has a unique nonzero principal maximal ideal $p\mathbb{Z}_p$ and $p$ is called a uniformizer of $\mathbb{Q}_p$. The unit group of $\mathbb{Z}_p$ is
$$\mathbb{Z}_p^{\times}= \left\{x\in\mathbb{Q}_p:\left|x\right|_p=1\right\}.$$
The residue class field $\mathbb{Z}_p/p\mathbb{Z}_p$ is a finite field with $p$ elements.\par
Let $n$ be a positive integer, and let $K$ be an extension field of $\mathbb{Q}_p$ of degree $n$. We fix some algebraic closure $\overline{\mathbb{Q}}_p$ of $\mathbb{Q}_p$ and view $K$ as a subfield of $\overline{\mathbb{Q}}_p$. Such $K$ exists. For example, let $K=\mathbb{Q}_p(\alpha)$ with $\alpha^n=p$. Because $X^n-p$ is an Eisenstein polynomial over $\mathbb{Q}_p$, it is irreducible over $\mathbb{Q}_p$, then $K$ has degree $n$ over $\mathbb{Q}_p$. The $p$-adic absolute value (or norm) $|\cdot|_p$ on $\mathbb{Q}_p$ can be extended uniquely to $K$, i.e., for $x\in K$, we have
$$\left|x\right|_p=\left| N_{K/\mathbb{Q}_p}(x)\right|_p^{\frac{1}{n}},$$
where $N_{K/\mathbb{Q}_p}$ is the norm map from $K$ to $\mathbb{Q}_p$. $K$ is complete with respect to $|\cdot|_p$.\par
Denote
$$\mathcal{O}_K=\left\{x\in K:\left|x\right|_p\le1\right\}.$$
$\mathcal{O}_K$ is also a discrete valuation ring, it has a unique nonzero principal maximal ideal $\pi\mathcal{O}_K$ and $\pi$ is called a uniformizer of $K$. $\mathcal{O}_K$ is a free $\mathbb{Z}_p$-module of rank $n$. $\mathcal{O}_K$ is compact and $K$ is locally compact. The unit group of $\mathcal{O}_K$ is
$$\mathcal{O}_K^{\times}= \left\{x\in K:\left|x\right|_p=1\right\}.$$
The residue class field $\mathcal{O}_K/\pi\mathcal{O}_K$ is a finite extension of $\mathbb{Z}_p/p\mathbb{Z}_p$. Call the positive integer
$$f=\left[\mathcal{O}_K/\pi\mathcal{O}_K:\mathbb{Z}_p/p\mathbb{Z}_p\right]$$
the residue field degree of $K/\mathbb{Q}_p$. As ideals in $\mathcal{O}_K$, we have $p\mathcal{O}_K=\pi ^e\mathcal{O}_K$. Call the positive integer $e$ the ramification index of $K/\mathbb{Q}_p$. We have $n=\left[K:\mathbb{Q}_p\right]=ef$. When $e=1$, the extension $K/\mathbb{Q}_p$ is unramified, and when $e=n$, $K/\mathbb{Q}_p$ is totally ramified. Each element $x$ of the multiplicative group $K^{\times}$ of nonzero elements of $K$ can be written uniquely as $x=u\pi^t$ with $u\in\mathcal{O}_K^{\times}$ and $t\in\mathbb{Z}$. We have $p=u\pi^e$ with $u\in\mathcal{O}_K^{\times}$, so $\left|\pi\right|_p=p^{-\frac{1}{e}}$.

\subsection{Norm and Orthogonal Basis}

Let $p$ be a prime. Let $V$ be a vector space over $\mathbb{Q}_p$. A norm on $V$ is a function
$$|\cdot|:V\rightarrow\mathbb{R}$$
such that

\begin{enumerate}
\item $\left|v\right|\ge0$ for any $v\in V$, and $\left|v\right|=0$ if and only if $v=0$;
\item $\left|xv\right|=\left|x\right|_{p}\cdot \left|v\right|$ for any $x\in\mathbb{Q}_p$ and $v\in V$;
\item $\left|v+w\right|\le\max{\left\{\left|v\right|,\left|w\right|\right\}}$ for any $v,w\in V$.
\end{enumerate}

If $|\cdot|$ is a norm on $V$, and if $\left|v\right|\ne \left|w\right|$ for $v,w\in V$, then we must have $\left|v+w\right|=\max{\{\left|v\right|,\left|w\right|\}}$. Weil (\cite{ref-5} page 26) proved the following proposition.

\begin{proposition}[\cite{ref-5}]
Let $V$ be a vector space over $\mathbb{Q}_p$ of finite dimension $n>0$, and let $\left|\cdot\right|$ be a norm on $V$. Then there is a decomposition $V=V_1+V_2+\cdots+V_n$ of $V$ into a direct sum of subspaces $V_i$ of dimension $1$, such that
$$\left|\sum_{i=1}^{n}{v_i}\right|=\max_{1\le i\le n}{\left|v_i\right|}$$
for any $v_i\in V_i$, $i=1,2,\dots,n$.
\end{proposition}

Weil proved the above proposition for finite-dimensional vector spaces over a $p$-field (commutative or not). For simplicity, we only consider the case $\mathbb{Q}_p$. Thus, we can define the orthogonal basis.

\begin{definition}[orthogonal basis \cite{ref-5}]
Let $V$ be a vector space over $\mathbb{Q}_p$ of finite dimension $n>0$, and let $|\cdot|$ be a norm on $V$. We call $\alpha_1,\alpha_2,\dots,\alpha_n$ an orthogonal basis of $V$ over $\mathbb{Q}_p$ if $V$ can be decomposed into the direct sum of $n$ $1$-dimensional subspaces $V_i={\rm Span}_{\mathbb{Q}_p}(\alpha_i)$ $(1\le i\le n)$, such that
$$\left|\sum_{i=1}^{n}{v_i}\right|=\max_{1\le i\le n}{\left|v_i\right|}$$
for any $v_i\in V_i$, $i=1,2,\dots,n$.
Two subspaces $U$, $W$ of $V$ is said to be orthogonal if the sum $U+W$ is a direct sum and it holds that $\left|u+w\right|=\max\left\{\left|u\right|, \left|w\right|\right\}$ for all $u\in U$, $w\in W$.
\end{definition}

We have the following immediate proposition and lemma.

\begin{proposition}[\cite {ref-2}]\label{pr-2.3}
Let $V$ be a vector space over $\mathbb{Q}_p$ of finite dimension $n>0$,  and let $|\cdot|$ be a norm on $V$. Let $\alpha_1,\alpha_2,\dots,\alpha_n$ be a basis of $V$ over $\mathbb{Q}_p$. If
$$\left\{\left|v_i\right|:v_i\in\mathbb{Q}_p\cdot\alpha_i\right\}\bigcap\left\{\left|v_j\right|:v_j\in\mathbb{Q}_p\cdot\alpha_j\right\}=\left\{0\right\}$$
for any $i,j=1,2,\dots,n$, $i\ne j$, then $\alpha_1,\alpha_2,\dots,\alpha_n$ is an orthogonal basis of $V$ over $\mathbb{Q}_p$.
\end{proposition}

\begin{lemma}[\cite{ref-3}]\label{le-2.4}
Let $V$ be a vector space over $\mathbb{Q}_p$ of finite dimension $n>0$, and let $|\cdot|$ be a norm on $V$. Let $\alpha_{1},\alpha_{2},\dots,\alpha_{n}$ be $\mathbb{Q}_p$-linearly independent vectors of $V$. Then $\alpha_{1},\alpha_{2},\dots,\alpha_{n}$ is an orthogonal basis of $V$ if and only if it holds that
$$\left|\sum^{n}_{i=1}a_{i}\alpha_{i}\right|=\max_{1\le i\le n}\left|a_{i}\alpha_{i}\right|,$$
where one of the $a_{1},a_{2},\dots,a_{n}$ is $1$ and the others are in $\mathbb{Z}_p$.
\end{lemma}

\subsection{$p$-adic Lattice}

We recall the definition of a $p$-adic lattice.

\begin{definition}[$p$-adic lattice \cite{ref-3}]
Let $V$ be a vector space over $\mathbb{Q}_p$ of finite dimension $n>0$. Let $m$ be a positive integer with $1\le m\le n$. Let $\alpha_1,\dots,\alpha_m\in V$ be $\mathbb{Q}_p$-linearly independent vectors. A $p$-adic lattice in $V$ is the set
$$\mathcal{L}(\alpha_1,\dots,\alpha_m):=\left\{\sum^{m}_{i=1}{a_i\alpha_i}:a_i\in\mathbb{Z}_p,1\le i\le m\right\}$$
of all $\mathbb{Z}_p$-linear combinations of $\alpha_1,\dots,\alpha_m$. The sequence of vectors $\alpha_1,\dots,\alpha_m$ is called a basis of the lattice $\mathcal{L}(\alpha_1,\dots,\alpha_m)$. The integers $m$ and $n$ are called the rank and dimension of the lattice, respectively. When $n=m$, we say that the lattice is of full rank.
\end{definition}

We can also define the orthogonal basis of a $p$-adic lattice.

\begin{definition}[orthogonal basis of a $p$-adic lattice \cite{ref-3}]
Let $V$ be nontrivial finite-dimensional vector space over $\mathbb{Q}_p$, and let $\alpha_1,\dots,\alpha_n\in V$. If $\alpha_1,\dots,\alpha_n\in V$ form an orthogonal basis of the vector space ${\rm Span_{\mathbb{Q}_p}}(\alpha_1,\dots,\alpha_n)$, then we say $\alpha_1,\dots,\alpha_n$ is an orthogonal basis of the lattice $\mathcal{L}(\alpha_1,\dots,\alpha_n)$.
\end{definition}

\subsection{LVP and CVP}

Deng et al. \cite{ref-1',ref-1} introduced two new computational problems in $p$-adic lattices. They are the Longest Vector Problem (LVP) and the Closest Vector Problem (CVP). We review them briefly.

\begin{definition}[\cite{ref-1',ref-1}]
Let $\mathcal{L}=\mathcal{L}(\alpha_1,\alpha_2,\dots,\alpha_m)$ be a lattice in $V$. We define recursively a sequence of positive real numbers $\lambda_1,\lambda_2,\lambda_3,\dots$ as follows.
$$\lambda_1=\max_{1\le i\le m}{\left|\alpha_i\right|_p},$$
$$\lambda_{j+1}=\max{\left\{\left|v\right|_p:v\in\mathcal{L},\left|v\right|_p<\lambda_j\right\}} \mbox{ for } j\ge1.$$
\end{definition}

We have $\lambda_1>\lambda_2>\lambda_3>\dots$ and $\lim_{j\rightarrow\infty}\lambda_j=0$. The Longest Vector Problem is defined as follows.

\begin{definition}[\cite{ref-1',ref-1}]
Given a lattice $\mathcal{L}=\mathcal{L}(\alpha_1,\alpha_2,\dots,\alpha_m)$ in $V$, the Longest Vector Problem is to find a lattice vector $v\in\mathcal{L}$ such that $\left|v\right|_p=\lambda_2$.
\end{definition}

The Closest Vector Problem is defined as follows.

\begin{definition}[\cite{ref-1',ref-1}]
Let $\mathcal{L}=\mathcal{L}(\alpha_1,\alpha_2,\dots,\alpha_m)$ be a lattice in $V$ and let $t\in V$ be a target vector. The Closest Vector Problem is to find a lattice vector $v\in\mathcal{L}$ such that $\left|t-v\right|_p=\min_{w\in\mathcal{L}}{\left|t-w\right|_p}$.
\end{definition}

Deng et al. \cite{ref-1} provided deterministic exponential time algorithms to solve the LVP and the CVP. Additionally, Deng et al. \cite{ref-2} presented deterministic polynomial time algorithms for solving the LVP and the CVP specifically with the help of orthogonal bases.

\subsection{Successive Maxima}

The successive maxima in $p$-adic lattices is the analogue of the successive minima in Euclidean lattices. It refers to the lengths of the shortest yet orthogonal vectors in a given $p$-adic lattice.

\begin{definition}[successive maxima \cite{ref-3.5}]
Let $V$ be a vector space over $\mathbb{Q}_p$, and let $|\cdot|$ be a norm on $V$. Let $\mathcal{L}$ be a lattice of rank $n$ in $V$. Let $\alpha_1,\alpha_2,\dots,\alpha_n$ be an orthogonal basis of $\mathcal{L}$ such that $\left|\alpha_{1}\right|\ge\left|\alpha_{2}\right|\ge\cdots\ge\left|\alpha_{n}\right|$. The $i$th successive maxima of $\mathcal{L}$ respect to norm $|\cdot|$ is
$$\tilde{\lambda}_i(\mathcal{L}):=\left|\alpha_i\right|.$$
\end{definition}

To ensure that this is well defined, Zhang et al. \cite{ref-3.5} proved that the sorted norm sequence of orthogonal bases of a lattice is unique.

\begin{theorem}[\cite{ref-3.5}]\label{th-2.10}
Let $V$ be a vector space over $\mathbb{Q}_p$, and let $|\cdot|$ be a norm on $V$. Let $\mathcal{L}$ be a lattice of rank $n$ in $V$. Suppose that $\alpha_{1},\alpha_2,\dots,\alpha_{n}$ and $\beta_{1},\beta_2,\dots,\beta_{n}$ are two orthogonal bases of $\mathcal{L}$ such that $\left|\alpha_{1}\right|\ge\left|\alpha_{2}\right|\ge\cdots\ge\left|\alpha_{n}\right|$ and $\left|\beta_{1}\right|\ge\left|\beta_{2}\right|\ge\cdots\ge\left|\beta_{n}\right|$. Then we have $\left|\alpha_i\right|=\left|\beta_i\right|$ for $1\le i\le n$.
\end{theorem}

\subsection{Signature Scheme and Public-key Encryption Cryptosystem}

In this subsection, we briefly review the $p$-adic signature scheme and public-key encryption cryptosystem presented in \cite{ref-2}.

\subsubsection{The $p$-adic Signature Scheme}

{\bf Key Generation:} We first choose a totally ramified $K$ of degree $n$ over $\mathbb{Q}_p$, i.e., choose an Eisenstein polynomial
$$f(x)=x^n+f_1x^{n-1}+\cdots+f_{n-1}x+f_n\in\mathbb{Z}_p[x]$$
satisfying $|f_n|_p=p^{-1}$ and $|f_i|_p<1$ for $1\le i\le n-1$. Let $\theta$ be a root of $f(x)= 0$. Choose another $\zeta\in\mathcal{O}_K=\mathbb{Z}_p[\theta]$ such that $\mathbb{Z}_p[\theta]=\mathbb{Z}_p[\zeta]$. Then $K=\mathbb{Q}_p(\zeta)$. Let $F(x)\in\mathbb{Z}_p[x]$ be the minimum polynomial of $\zeta$ over $\mathbb{Q}_p$, which is also monic and of degree $n$. Choose $n$ non-negative integers $j_i\in\mathbb{Z}$ such that $j_i\pmod n$ $(1\le i\le n)$ are distinct and $j_1=0$. Set $\alpha_i=\theta^{j_i}$ $(1\le i\le n)$. Then $\alpha_1=1,\alpha_2,\dots,\alpha_n$ is an orthogonal basis. All elements of $\mathcal{O}_K$ should be expressed as polynomials in $\zeta$ of degree less than $n$ with coefficients in $\mathbb{Z}_p$ and $\zeta$ is just a formal symbol.\par
Choose a positive integer $m\le n$. Choose a matrix $A\in {\rm GL}_m(\mathbb{Z}_p)$. Put
$$\left( \begin{array}{c} \beta_{1} \\ \beta_{2} \\ \vdots \\ \beta_{m} \end{array} \right)=A\left( \begin{array}{c} \alpha_{1} \\ \alpha_{2} \\ \vdots \\ \alpha_{m} \end{array} \right)$$
such that $\left|\beta_1\right|_p=\left|\beta_2\right|_p=\cdots=\left|\beta_m\right|_p=1$. Set
\begin{equation*}
\begin{split}
\mathcal{L}
&=\mathbb{Z}_p\cdot\beta_{1}+\mathbb{Z}_p\cdot\beta_{2}+\cdots+\mathbb{Z}_p\cdot\beta_{m}\\
&=\mathbb{Z}_p\cdot\alpha_{1}+\mathbb{Z}_p\cdot\alpha_{2}+\cdots+\mathbb{Z}_p\cdot\alpha_{m}.
\end{split}
\end{equation*}
We need a hash function
$$H:\{0,1\}^{*}\rightarrow W:=\left\{x\big|x\in K-\mathcal{L},\left|x\right|_{p}=1\right\},$$
\par
{\bf Public key} is set to be: $(F(X),H,(\beta_{1},\beta_{2},\dots,\beta_{m}))$.\par
{\bf Private key} is set to be: $(\alpha_1,\alpha_2,\dots,\alpha_m,\alpha_{m+1},\dots,\alpha_n)$.\par
\noindent {\bf Signing algorithm:} For any message $M\in\{0, 1\}^*$, choose a random number $r$ of fixed length, say, $r\in\{0, 1\}^{256}$. Compute
$$t=H(M\|r).$$
Using the orthogonal basis $(\alpha_1,\alpha_2,\dots,\alpha_m,\alpha_{m+1},\dots,\alpha_n)$, Bob computes a lattice vector $v\in\mathcal{L}$ which is closest to $t$. If the minimum value of the distance from $t$ to $\mathcal{L}$ is strictly less than $1$, then output the signature $(r,v)$. Otherwise, choose a new $r\in\{0,1\}^{256}$ until the minimum value of the distance from $t$ to $\mathcal{L}$ is strictly less than $1$.\par
\noindent {\bf Verification algorithm:} The signature is valid if and only if $t=H(M\|r)$, $v\in\mathcal{L}$ and $\left|t-v\right|_p<1$.\par
The correctness is obvious. For the efficiency, we can prove that  for a fixed $f(x)$, the probability of $\mathbb{Z}_p[\theta]=\mathbb{Z}_p[\zeta]$ is $1-\frac{1}{p}$. Also, we can always get a valid signature with just one $r$. See Lemma \ref{le-A.1} and Lemma \ref{le-A.2} in Appendix \ref{se-A}.

\subsubsection{The $p$-adic Public-key Encryption Cryptosystem}

{\bf Key Generation:} The same as in the signature scheme. Additionally, choose a real number $\delta\ge0$.\par
{\bf Public key} is set to be: $(F(X),\delta,(\beta_{1},\beta_{2},\dots,\beta_{m}))$.\par
{\bf Private key} is set to be: $(A,(\alpha_1,\alpha_2,\dots,\alpha_n))$.\par
\noindent {\bf Encryption:} For any plaintext $(a_1,a_2,\dots,a_m)\in\{0,1,\dots,p-1\}^m$, Alice first chooses randomly $r\in K$ with $\left|r\right|_p<p^{-\delta}$, computes the ciphertext
$$C=a_1\beta_1+a_2\beta_2+\cdots+ a_m\beta_m+r\in K$$
and sends $C$ to Bob.\par
\noindent {\bf Decryption:} When Bob receives the ciphertext $C$, he computes a lattice vector $v\in\mathcal{L}$ which is closest to $C$. Write
$$v=b_1\alpha_1+b_2\alpha_2+\cdots+b_m\alpha_m,$$
where $b_i\in\mathbb{Z}_p$, then the plaintext is
$$(b_1,b_2,\dots,b_m)\cdot A^{-1}\pmod p.$$\par
For the correctness, the following theorem is proved in \cite{ref-2}:
\begin{theorem}
The decryption is correct if it holds that $j_i\le\delta n$ for $1\le i\le m$.
\end{theorem}

\section{Orthogonalization in Local Fields}\label{se-3}

In \cite{ref-1',ref-1}, Deng et al. provided deterministic algorithms to solve the LVP and the CVP. The running time of the LVP algorithm is exponential in the dimension of the input lattice. On the other hand, there are no known deterministic polynomial time algorithms to solve the LVP or the CVP. However, we can take advantage of the special structure of local fields to simplify these algorithms.\par
As mentioned in \cite{ref-2} and \cite{ref-3.5}, the LVP and the CVP can be solved in polynomial time with the help of an orthogonal basis. Therefore, we consider orthogonalization algorithms directly. In totally ramified fields, the simplified orthogonalization algorithm takes $O(pm^2)$ many $p$-adic absolute value computations of elements of $K$, where $m$ is the rank of the input lattice. It is a polynomial time algorithm when $p$ is a polynomial in $m$.

\subsection{Orthogonalization in Totally Ramified Fields}

Let $K$ be an extension field of $\mathbb{Q}_p$ of degree $n$. Let $\pi$ be a uniformizer of $K$. If $K$ is a totally ramified extension field over $\mathbb{Q}_p$, then $\left|\pi\right|_p=p^{-\frac{1}{n}}$ and
$$\mathcal{O}_K=\mathbb{Z}_p[\pi]=\mathcal{L}(1,\pi,\pi^2,\dots,\pi^{n-1}).$$
By Proposition 3.4 in \cite{ref-2}, $1,\pi,\pi^2,\dots,\pi^{n-1}$ is an orthogonal basis of $\mathcal{O}_K$. The crucial point is that this orthogonal basis satisfies
$$\left|1\right|_p>\left|\pi\right|_p>\left|\pi^2\right|_p>\cdots>\left|\pi^{n-1}\right|_p>\left|p\cdot1\right|_p.$$
If a lattice has an orthogonal basis with the above property, then the orthogonalization process of this lattice will be extremely simplified, as shown in the following theorems.
\begin{lemma}\label{le-4.1}
Let $\mathcal{L}=\mathcal{L}(\alpha_1,\alpha_2,\dots,\alpha_m)$ be a lattice in $K$ $(m\ge2)$, such that $\left|\alpha_1\right|_p>\left|\alpha_2\right|_p\ge\cdots\ge\left|\alpha_m\right|_p$. Then $\lambda_2=\max{\left\{\left|p\alpha_1\right|_p,\left|\alpha_2\right|_p\right\}}$.
\end{lemma}
\begin{proof}
Let $v=\sum_{i=1}^{m}{a_i\alpha_i}$ be a lattice vector in $\mathcal{L}$. If $a_1\in\mathbb{Z}_p-p\mathbb{Z}_p$, then $\left|v\right|_p=\left|\alpha_1\right|_p=\lambda_1$. If $a_1\in p\mathbb{Z}_p$, then
$$\left|v\right|_p\le\max{\left\{\left|a_1\alpha_1\right|_p,\left|\sum_{i=2}^{m}{a_i\alpha_i}\right|_p\right\}}\le\max{\left\{\left|p\alpha_1\right|_p,\left|\alpha_2\right|_p\right\}}.$$
Therefore, there is no lattice vector $v\in\mathcal{L}$ such that
$$\max{\left\{\left|p\alpha_1\right|_p,\left|\alpha_1\right|_p\right\}}<\left|v\right|_p<\lambda_1.$$
On the other hand, there exists lattice vector in $\mathcal{L}$ with length $\left|p\alpha_1\right|_p$ or $\left|\alpha_2\right|_p$, namely $p\alpha_1$ and $\alpha_2$. Hence $\lambda_2=\max{\left\{\left|p\alpha_1\right|_p,\left|\alpha_2\right|_p\right\}}$.
\end{proof}

\begin{theorem}\label{th-4.2}
Let $\mathcal{L}=\mathcal{L}(\beta_1,\beta_2,\dots,\beta_m)$ be a lattice in $K$ $(m\ge2)$. Suppose that $\mathcal{L}$ has a basis $\alpha_1,\alpha_2,\dots,\alpha_m$ such that $\left|\alpha_1\right|_p>\left|\alpha_2\right|_p>\cdots>\left|\alpha_m\right|_p>\left|p\alpha_1\right|_p$. Then there is a deterministic algorithm to find the value $\lambda_2$ and a lattice vector $v\in\mathcal{L}$ such that $\left|v\right|_p=\lambda_2$. The algorithm takes $O(pm)$ many $p$-adic absolute value computations of elements of $K$.
\end{theorem}
\begin{proof}
Since
$$\left\{\left|v_i\right|:v_i\in\mathbb{Q}_p\cdot\alpha_i\right\}=\{0\}\cup p^{\mathbb{Z}}\cdot\left|\alpha_i\right|$$
and
$$\left|\alpha_1\right|_p>\left|\alpha_2\right|_p>\cdots>\left|\alpha_m\right|_p>\left|p\alpha_1\right|_p,$$
we have
$$\left\{\left|v_i\right|:v_i\in\mathbb{Q}_p\cdot\alpha_i\right\}\bigcap\left\{\left|v_j\right|:v_j\in\mathbb{Q}_p\cdot\alpha_j\right\}=\left\{0\right\}$$
for any $i,j=1,2,\dots,n$, $i\ne j$. Therefore, by Proposition \ref{pr-2.3}, $\alpha_1,\alpha_2,\dots,\alpha_m$ is an orthogonal basis of $\mathcal{L}$. Let $A=(a_{ij})\in {\rm GL}_m(\mathbb{Z}_p)$ be an invertible matrix over $\mathbb{Z}_p$ such that
\begin{equation*}
\left(\begin{array}{c}\beta_1\\ \beta_2\\ \vdots\\ \beta_m \end{array}\right)
=A\left(\begin{array}{c}\alpha_1\\ \alpha_2\\ \vdots\\ \alpha_m \end{array}\right).
\end{equation*}
Then we have
\begin{equation*}
\begin{array}{c}
\beta_1=a_{11}\alpha_1+a_{12}\alpha_2+\cdots+a_{1m}\alpha_m\\
\beta_2=a_{21}\alpha_1+a_{22}\alpha_2+\cdots+a_{2m}\alpha_m\\
\vdots\\
\beta_m=a_{m1}\alpha_1+a_{m2}\alpha_2+\cdots+a_{mm}\alpha_m.
\end{array}
\end{equation*}
If $a_{i1}\in\mathbb{Z}_p-p\mathbb{Z}_p$, then we have $\left|\beta_i\right|_p=\left|\alpha_1\right|_p=\lambda_1$. If $a_{i1}\in p\mathbb{Z}_p$, then we have $\left|\beta_i\right|_p\le\left|\alpha_2\right|_p<\lambda_1$. Since the longest vector of $\mathcal{L}$ is $\alpha_1$, there must be at least one $\beta_i$ such that $\left|\beta_i\right|_p=\lambda_1$. Hence, we may assume that for $1\le i\le l$, $a_{i1}\in\mathbb{Z}_p-p\mathbb{Z}_p$ and for $l+1\le i\le m$, $a_{i1}\in p\mathbb{Z}_p$. If $l=1$, then
$$\lambda_2=\max{\left\{\left|\beta_2\right|_p,\left|\beta_3\right|_p,\dots,\left|\beta_m\right|_p\right\}}.$$
We can find the second longest vector $v$ directly from $\beta_2,\beta_3,\dots,\beta_m$. Now we assume that $l\ge2$. Then for each $2\le i\le l$, there is an integer $c_i\in\{0,1,\dots,p-1\}$ such that $a_{i1}-c_ia_{11}\in p\mathbb{Z}_p$ and hence
$$\left|\beta_i-c_i\beta_1\right|_p\le\left|\alpha_2\right|_p<\lambda_1.$$
Let $\beta_i^{\prime}=\beta_i-c_i\beta_1$ for $2\le i\le l$ and $\beta_i^{\prime}=\beta_i$ for the remaining $1\le i\le m$. Then
$$\mathcal{L}(\beta_1,\beta_2,\dots,\beta_m)=\mathcal{L}(\beta_1^{\prime},\beta_2^{\prime},\dots,\beta_m^{\prime}),$$
and
$$\left|\beta_1^{\prime}\right|_p=\left|\alpha_1\right|_p=\lambda_1>\max{\left\{\left|\beta_2^{\prime}\right|_p,\left|\beta_3^{\prime}\right|_p,\dots,\left|\beta_m^{\prime}\right|_p\right\}}.$$
Let $2\le s\le m$ be the subscript such that
$$\left|\beta_s^{\prime}\right|_p=\max{\left\{\left|\beta_2^{\prime}\right|_p,\left|\beta_3^{\prime}\right|_p,\dots,\left|\beta_m^{\prime}\right|_p\right\}}.$$
By Lemma \ref{le-4.1}, we have $\lambda_2=\max{\left\{\left|p\beta_1^{\prime}\right|_p,\left|\beta_s^{\prime}\right|_p\right\}}$ and $v=p\beta_1$ or $\beta_s^{\prime}$. Let $A^{\prime}=(a_{ij}^{\prime})\in {\rm GL}_m(\mathbb{Z}_p)$ be an invertible matrix over $\mathbb{Z}_p$ such that
\begin{equation*}
\left(\begin{array}{c}\beta_1^{\prime}\\ \beta_2^{\prime}\\ \vdots\\ \beta_m^{\prime} \end{array}\right)
=A^{\prime}\left(\begin{array}{c}\alpha_1\\ \alpha_2\\ \vdots\\ \alpha_m \end{array}\right).
\end{equation*}
Then we have
\begin{equation*}
\begin{array}{c}
\beta_1^{\prime}=a_{11}^{\prime}\alpha_1+a_{12}^{\prime}\alpha_2+\cdots+a_{1m}^{\prime}\alpha_m\\
\beta_2^{\prime}=a_{21}^{\prime}\alpha_1+a_{22}^{\prime}\alpha_2+\cdots+a_{2m}^{\prime}\alpha_m\\
\vdots\\
\beta_m^{\prime}=a_{m1}^{\prime}\alpha_1+a_{m2}^{\prime}\alpha_2+\cdots+a_{mm}^{\prime}\alpha_m,
\end{array}
\end{equation*}
where $a_{i1}^{\prime}=a_{i1}-c_ia_{11}\in p\mathbb{Z}_p$ for $2\le i\le m$.
If $\left|\beta_s^{\prime}\right|_p<\left|p\beta_1^{\prime}\right|_p$, then we have
$$\left|\beta_s^{\prime}\right|_p<\left|p\beta_1^{\prime}\right|_p=\left|p\alpha_1\right|_p<\left|\alpha_2\right|_p.$$
Hence $a_{i2}^{\prime}\in p\mathbb{Z}_p$ for $2\le i\le m$, which implies that $\det(A^{\prime})\in p\mathbb{Z}_p$ and $A^{\prime}$ is not invertible over $\mathbb{Z}_p$. This is a contradiction. Therefore, we must have $\left|\beta_s^{\prime}\right|_p\ge\left|p\beta_1^{\prime}\right|_p$, $\lambda_2=\left|\beta_s^{\prime}\right|_p$ and a second longest vector is $v=\beta_s^{\prime}$.
\end{proof}

We summarize the algorithm in above theorem as follows.


\begin{algorithm}
\caption{finding $\lambda_2$}\label{al-1}
\begin{algorithmic}[1]
\Require a lattice $\mathcal{L}=\mathcal{L}(\beta_1,\beta_2,\dots,\beta_m)$
\Ensure the value $\lambda_2$ and a lattice vector $v\in\mathcal{L}$ such that $\left|v\right|_p=\lambda_2$ 
\State rearrange $\beta_1,\beta_2,\dots,\beta_m$ such that  $\left|\beta_1\right|_p=\max{\left\{\left|\beta_1\right|_p,\left|\beta_2\right|_p,\dots,\left|\beta_m\right|_p\right\}}$
\State $\beta_1^{\prime}\Leftarrow\beta_1$
\For{$i=2$ to $m$}
 	\For{$j=0$ to $p-1$}
		\If{$\left|\beta_i-j\beta_1\right|_p<\left|\beta_1\right|_p$}
			\State $\beta_i^{\prime}\Leftarrow\beta_i-j\beta_1$
		\EndIf
	\EndFor
\EndFor
\State let $s$ be the subscript such that
$\left|\beta_s^{\prime}\right|_p=\max{\left\{\left|\beta_2^{\prime}\right|_p,\left|\beta_3^{\prime}\right|_p,\dots,\left|\beta_m^{\prime}\right|_p\right\}}$
\State\Return $\lambda_2\Leftarrow\left|\beta_s^{\prime}\right|_p$, $v\Leftarrow\beta_s^{\prime}$
\end{algorithmic}
\end{algorithm}

This algorithm takes at most $m+p(m-1)$ many $p$-adic absolute value computations of elements of $K$. Notice that this algorithm also computes a new basis $\beta_1^{\prime},\beta_2^{\prime},\dots,\beta_m^{\prime}$ of the lattice $\mathcal{L}$ such that
$$\left|\beta_1^{\prime}\right|_p=\left|\alpha_1\right|_p>\left|\alpha_{2}\right|_p\ge\max{\left\{\left|\beta_{2}^{\prime}\right|_p,\left|\beta_{3}^{\prime}\right|_p,\dots,\left|\beta_{m}^{\prime}\right|_p\right\}},$$
where the meanings of $\alpha_1$ and $\alpha_2$ are the same as those in Theorem \ref{th-4.2}. The following theorem shows that we can recursively run this algorithm to obtain an orthogonal basis of the lattice $\mathcal{L}$.

\begin{theorem}
Let $\mathcal{L}=\mathcal{L}(\beta_1,\beta_2,\dots,\beta_m)$ be a lattice in $K$ $(m\ge2)$. Suppose that $\mathcal{L}$ has a basis $\alpha_1,\alpha_2,\dots,\alpha_m$ such that $\left|\alpha_1\right|_p>\left|\alpha_2\right|_p>\cdots>\left|\alpha_m\right|_p>\left|p\alpha_1\right|_p$. Then there is a deterministic algorithm to find an orthogonal basis of $\mathcal{L}$. The algorithm takes $O(pm^2)$ many $p$-adic absolute value computations of elements of $K$.
\end{theorem}
\begin{proof}
Recursively run the algorithm in Theorem \ref{th-4.2}. Suppose that we have already obtained a basis $\beta_1^{\prime},\beta_2^{\prime},\dots,\beta_m^{\prime}$ of the lattice $\mathcal{L}$ such that
\begin{equation*}
\begin{split}
\left|\beta_1^{\prime}\right|_p&=\left|\alpha_1\right|_p>\left|\beta_2^{\prime}\right|_p=\left|\alpha_2\right|_p>\cdots>\left|\beta_l^{\prime}\right|_p=\left|\alpha_l\right|_p\\
&>\left|\alpha_{l+1}\right|_p\ge\max{\left\{\left|\beta_{l+1}^{\prime}\right|_p,\left|\beta_{l+2}^{\prime}\right|_p,\dots,\left|\beta_{m}^{\prime}\right|_p\right\}}.
\end{split}
\end{equation*}
for some $1\le l\le m$. If
$$\left|\alpha_{l+1}\right|_p>\max{\left\{\left|\beta_{l+1}^{\prime}\right|_p,\left|\beta_{l+2}^{\prime}\right|_p,\dots,\left|\beta_{m}^{\prime}\right|_p\right\}},$$
then for any lattice vector $v=\sum_{i=1}^{m}{a_i\beta_i^{\prime}}$, where $a_i\in\mathbb{Z}_p$, $1\le i\le m$, we have
$$\left|v\right|_p=\left|\sum_{i=1}^{m}{a_i\beta_i^{\prime}}\right|_p\ge\left|\beta_{s}^{\prime}\right|_p=\left|\alpha_{s}\right|_p\ge\left|\alpha_{l}\right|_p>\left|\alpha_{l+1}\right|_p$$
whenever there is some $a_s\in\mathbb{Z}_p-p\mathbb{Z}_p$ with subscript $1\le s\le l$, and
\begin{equation*}
\begin{split}
\left|v\right|_p&\le\max{\left\{\left|p\sum_{i=1}^{l}{\frac{a_i}{p}\beta_i^{\prime}}\right|,\left|\sum_{i=l+1}^{m}{a_i\beta_i^{\prime}}\right|_p\right\}}\\
&\le\max{\left\{\left|p\alpha_{1}\right|_p,\left|\beta_{l+1}^{\prime}\right|_p,\left|\beta_{l+2}^{\prime}\right|_p,\dots,\left|\beta_{m}^{\prime}\right|_p\right\}}<\left|\alpha_{l+1}\right|_p
\end{split}
\end{equation*}
whenever $a_i\in\mathbb{Z}_p$ for all $1\le i \le l$. Hence there is no lattice vector in the lattice $\mathcal{L}(\beta_1^{\prime},\beta_2^{\prime},\dots,\beta_m^{\prime})= \mathcal{L}(\alpha_1,\alpha_2,\dots,\alpha_m)$ with length $\left|\alpha_{l+1}\right|_p$. This is a contradiction. Thus, we must have
$$\left|\alpha_{l+1}\right|_p=\max{\left\{\left|\beta_{l+1}^{\prime}\right|_p,\left|\beta_{l+2}^{\prime}\right|_p,\dots,\left|\beta_{m}^{\prime}\right|_p\right\}}.$$
Run the algorithm with lattice $\mathcal{L}(\beta_{l+1}^{\prime},\beta_{l+2}^{\prime},\dots,\beta_{m}^{\prime})$, we obtain a basis
$$\beta_1^{\prime},\beta_2^{\prime},\dots,\beta_l^{\prime},\beta_{l+1}^{\prime\prime},\dots,\beta_{m}^{\prime\prime}$$
of the lattice $\mathcal{L}$ such that
\begin{equation*}
\begin{split}
\left|\beta_1^{\prime}\right|_p&=\left|\alpha_1\right|_p>\left|\beta_2^{\prime}\right|_p=\left|\alpha_2\right|_p>\cdots>\left|\beta_l^{\prime}\right|_p=\left|\alpha_l\right|_p\\
&>\left|\beta_{l+1}^{\prime\prime}\right|_p=\left|\alpha_{l+1}\right|_p>\left|\alpha_{l+2}\right|_p\ge\max{\left\{\left|\beta_{l+2}^{\prime\prime}\right|_p,\left|\beta_{l+3}^{\prime\prime}\right|_p,\dots,\left|\beta_{m}^{\prime\prime}\right|_p\right\}}.
\end{split}
\end{equation*}
Therefore, by induction, we finally obtain a basis $\gamma_1,\gamma_2,\dots,\gamma_m$ of the lattice $\mathcal{L}$ such that
$$\left|\gamma_1\right|_p=\left|\alpha_1\right|_p>\left|\gamma_2\right|_p=\left|\alpha_2\right|_p>\cdots>\left|\gamma_m\right|_p=\left|\alpha_m\right|_p.$$
By Proposition \ref{pr-2.3}, it is an orthogonal basis of $\mathcal{L}$.
\end{proof}

Let $A$ be an algorithm having the same input and procedure as Algorithm \ref{al-1}, but outputting the new basis, i.e., $A(\beta_1,\beta_2,\dots,\beta_m)=\beta_1^{\prime},\beta_2^{\prime},\dots,\beta_m^{\prime}$. We summarize the algorithm in above theorem as follows.


\begin{algorithm}
\caption{finding an orthogonal basis}\label{al-2}
\begin{algorithmic}[1]
\Require a lattice $\mathcal{L}=\mathcal{L}(\beta_1,\beta_2,\dots,\beta_m)$
\Ensure an orthogonal basis of $\mathcal{L}(\beta_1,\beta_2,\dots,\beta_m)$
\State $\beta_1^{\prime},\beta_2^{\prime},\dots,\beta_m^{\prime}\Leftarrow\beta_1,\beta_2,\dots,\beta_m$
\For{$i=1$ to $m-1$}
	\State $\beta_{i}^{\prime},\beta_{i+1}^{\prime},\dots,\beta_m^{\prime}\Leftarrow A(\beta_{i}^{\prime},\beta_{i+1}^{\prime},\dots,\beta_m^{\prime})$
\EndFor
\State\Return $\beta_1^{\prime},\beta_2^{\prime},\dots,\beta_m^{\prime}$
\end{algorithmic}
\end{algorithm}

This algorithm calls the subroutine $A$ for $m-1$ times, hence takes at most $m(m-1)+p(m-1)^2$ many $p$-adic absolute value computations of elements of $K$.

\subsection{LVP in General Local Fields}\label{se-4.1.2}

Let $K$ be an extension field of of $\mathbb{Q}_p$ degree $n=ef$, where $e$ is the ramification index and $f$ is the residue degree. Let $\pi$ be a uniformizer of $K$, then $\left|\pi\right|_p=p^{-\frac{1}{e}}$. According to Proposition 6.4.5 in \cite{ref-4.2}, we can choose $f$ units $s_1,s_2,\dots,s_f\in\mathcal{O}_K$ such that
$$\mathcal{O}_K=\mathcal{L}(s_i\pi^j)_{1\le i\le f, \mbox{ } 0\le j\le e-1}.$$
This basis satisfies
\begin{equation*}
\begin{split}
&\left|s_1\right|_p=\left|s_2\right|_p=\cdots=\left|s_f\right|_p\\
>&\left|s_1\pi\right|_p=\left|s_2\pi\right|_p=\cdots=\left|s_f\pi\right|_p\\
&\cdots\\
>&\left|s_1\pi^{e-1}\right|_p=\left|s_2\pi^{e-1}\right|_p=\cdots=\left|s_f\pi^{e-1}\right|_p\\
>&\left|p\cdot s_1\right|_p.
\end{split}
\end{equation*}
If a lattice has an orthogonal basis with the above property, then the LVP algorithm of this lattice can be simplified, though still to be an exponential time algorithm, as shown in Theorem \ref{th-4.5}.\par
Before presenting this theorem, we need a lemma.

\begin{lemma}\label{le-4.4}
Let $V$ be a vector space over $\mathbb{Q}_p$ of finite dimension $n>0$, and let $|\cdot|$ be a norm on $V$. Let $\alpha_{1},\alpha_{2},\dots,\alpha_{n}$ be a basis of $V$ over $\mathbb{Q}_p$ such that $\left|\alpha_{1}\right|=\left|\alpha_{2}\right|=\cdots=\left|\alpha_{n}\right|$, which is denoted by $\lambda_{1}$. Then $\alpha_{1},\alpha_{2},\dots,\alpha_{n}$ is an orthogonal basis of $V$ over $\mathbb{Q}_p$ if and only if
$$\left|\sum^{n}_{i=1}a_{i}\alpha_{i}\right|=\lambda_{1}$$
for all
$$(a_1,a_2,\dots,a_n)\in\{0,1,\dots,p-1\}^n-(0,0,\dots,0).$$
\end{lemma}
\begin{proof}
Assume $\alpha_{1},\alpha_{2},\dots,\alpha_{n}$ is an orthogonal basis of $V$ over $\mathbb{Q}_p$. Then 
\begin{equation*}
\left|\sum^{n}_{i=1}a_{i}\alpha_{i}\right|=\max_{1\le i\le n}\left| a_{i}\alpha_{i}\right|=\lambda_{1}\cdot\max_{1\le i\le n}\left| a_{i}\right|=\lambda_1
\end{equation*}
for all
$$(a_1,a_2,\dots,a_n)\in\{0,1,\dots,p-1\}^n-(0,0,\dots,0).$$
Conversely, assume that $\left|\sum^{n}_{i=1}a_{i}\alpha_{i}\right|=\lambda_{1}$ for all
$$(a_1,a_2,\dots,a_n)\in\{0,1,\dots,p-1\}^n-(0,0,\dots,0).$$
If $\alpha_{1},\alpha_{2},\dots,\alpha_{n}$ is not an orthogonal basis of $V$ over $\mathbb{Q}_p$, then by Lemma \ref{le-2.4}, there is $a_{i}\in\mathbb{Z}_p$, $1\le i\le n$, and at least one $a_{i}=1$, such that 
\begin{equation*}
\left|\sum^{n}_{i=1}a_{i}\alpha_{i}\right|<\max_{1\le i\le n}\left| a_{i}\alpha_{i}\right|=\lambda_{1}.
\end{equation*}
Let $a_i^{\prime}=a_i \mod p$, $1\le i\le n$, then
$$(a_1^{\prime},a_2^{\prime},\dots,a_n^{\prime})\in\{0,1,\dots,p-1\}^n-(0,0,\dots,0)$$
and
$$\left|\sum^{n}_{i=1}a_{i}^{\prime}\alpha_{i}\right|=\left|\sum^{n}_{i=1}a_{i}\alpha_{i}\right|<\lambda_{1}.$$
This is a contradiction. Hence $\alpha_{1},\alpha_{2},\dots,\alpha_{n}$ is an orthogonal basis of $V$ over $\mathbb{Q}_p$.
\end{proof}

Now we prove our theorem.

\begin{theorem}\label{th-4.5}
Let $\mathcal{L}=\mathcal{L}(\beta_1,\beta_2,\dots,\beta_m)$ be a lattice in $K$. Let $f<m$ be a positive integer. Suppose that $\mathcal{L}$ has an orthogonal basis $\alpha_1,\alpha_2,\dots,\alpha_m$ such that
$$\left|\alpha_1\right|_p=\left|\alpha_2\right|_p=\cdots=\left|\alpha_f\right|_p>\left|\alpha_{f+1}\right|_p\ge\cdots\ge\left|\alpha_m\right|_p>\left|p\alpha_1\right|_p.$$
Then there is a deterministic exponential time algorithm to find the value $\lambda_2$ and a lattice vector $v\in\mathcal{L}$ such that $\left|v\right|_p=\lambda_2$. The algorithm takes $O(mp^f)$ many $p$-adic absolute value computations of elements of $K$.
\end{theorem}
\begin{proof}
We may assume that for $1\le i\le l$, $\left|\beta_i\right|_p=\lambda_1$ and for $l+1\le i\le m$, $\left|\beta_i\right|_p<\lambda_1$. By Theorem \ref{th-2.10}, we must have $l\ge f$.\par
Suppose that $1\le k<l$ and $\beta_1,\beta_2,\dots,\beta_k$ is an orthogonal basis of $\mathcal{L}(\beta_1,\beta_2,\dots,\beta_k)$. According to Lemma \ref{le-4.4}, $\beta_1,\beta_2,\dots,\beta_k,\beta_{k+1}$ is an orthogonal basis of $\mathcal{L}(\beta_1,\beta_2,\dots,\beta_k,\beta_{k+1})$ if and only if
$$\left|\sum^{k+1}_{i=1}b_{i}\beta_{i}\right|_p=\lambda_{1}$$
for all
$$(b_1,b_2,\dots,b_{k+1})\in\{0,1,\dots,p-1\}^{k+1}-(0,0,\dots,0).$$
Since $\beta_1,\beta_2,\dots,\beta_k$ is an orthogonal basis of $\mathcal{L}(\beta_1,\beta_2,\dots,\beta_k)$, the case $b_{k+1}=0$ is already held. If $b_{k+1}\ne0$, then, by multiplying $\left|b^{-1}_{k+1}\right|_p=1$ to the left side of the equation, we may assume that $b_{k+1}=1$.\par
Therefore, we can check whether $\beta_1,\beta_2,\dots,\beta_k,\beta_{k+1}$ is an orthogonal basis of $\mathcal{L}(\beta_1,\beta_2,\dots,\beta_k,\beta_{k+1})$ by checking whether
$$\left|\beta_{k+1}-\sum_{i=1}^{k}{b_i\beta_i}\right|_p=\lambda_1$$
holds for every $b_1,b_2,\dots,b_k\in\{0,1,\dots,p-1\}$. If this condition holds, then we recursively check $\beta_1,\beta_2,\dots,\beta_{k+1},\beta_{k+2}$. Otherwise, we check $\beta_1,\beta_2,\dots,\beta_{k},\beta_{j}$ for $k+2\le j\le l$ until we find a larger orthogonal basis or all $j$'s are checked.\par
By Theorem \ref{th-2.10}, the number of maximum orthogonal vectors among $\beta_1,\beta_2,\dots,\beta_l$ is $f$, hence the above check can be done within $p^f$ many $p$-adic absolute value computations of elements of $K$.\par
We may assume that $\beta_1,\beta_2,\dots,\beta_f$ are maximum orthogonal vectors among $\beta_1,\beta_2,\dots,\beta_l$. Then for every $f+1\le j\le l$, there is $b_{j1},b_{j2},\dots,b_{jk}\in\{0,1,\dots,p-1\}$ such that $$\left|\beta_{j}-\sum_{i=1}^{k}{b_{ji}\beta_i}\right|_p<\lambda_1.$$
This can also be done within $p^f$ many $p$-adic absolute value computations of elements of $K$. Let
$$\beta_{j}^{\prime}=\beta_{j}-\sum_{i=1}^{k}{b_{ji}\beta_i}$$
for $f+1\le i\le l$ and $\beta_i^{\prime}=\beta_i$ for the remaining $1\le i\le m$. Then
$$\mathcal{L}(\beta_1,\beta_2,\dots,\beta_m)=\mathcal{L}(\beta_1^{\prime},\beta_2^{\prime},\dots,\beta_m^{\prime}),$$
and
$$\left|\beta_1^{\prime}\right|_p=\left|\beta_2^{\prime}\right|_p=\cdots=\left|\beta_f^{\prime}\right|_p=\lambda_1>\max{\left\{\left|\beta_{f+1}^{\prime}\right|_p,\left|\beta_{f+2}^{\prime}\right|_p,\dots,\left|\beta_m^{\prime}\right|_p\right\}}.$$
Let $f+1\le s\le m$ be the subscript such that
$$\left|\beta_s^{\prime}\right|_p=\max{\left\{\left|\beta_{f+1}^{\prime}\right|_p,\left|\beta_{f+2}^{\prime}\right|_p,\dots,\left|\beta_m^{\prime}\right|_p\right\}}.$$
Since $\left|\beta_s^{\prime}\right|_p\ge\left|\alpha_m\right|_p>\left|p\alpha_1\right|_p=\left|p\beta_1\right|_p$, by a generalization of Lemma \ref{le-4.1}, we conclude that $\lambda_2=\left|\beta_s^{\prime}\right|_p$ and a second longest vector is $v=\beta_s^{\prime}$.
\end{proof}

We summarize the algorithm in above theorem as Algorithm \ref{al-3}.


\begin{algorithm}[H]
\caption{finding $\lambda_2$}\label{al-3}
\begin{algorithmic}[1]
\Require a lattice $\mathcal{L}=\mathcal{L}(\beta_1,\beta_2,\dots,\beta_m)$
\Ensure the value $\lambda_2$ and a lattice vector $v\in\mathcal{L}$ such that $\left|v\right|_p=\lambda_2$ 
\State rearrange $\beta_1,\beta_2,\dots,\beta_m$ such that  $\left|\beta_1\right|_p=\max{\left\{\left|\beta_1\right|_p,\left|\beta_2\right|_p,\dots,\left|\beta_m\right|_p\right\}}$
\State $\beta_1^{\prime}\Leftarrow\beta_1$, $L\Leftarrow[\beta_1]$ (a list, indices begin with 1)
\For{$i=2$ to $m$}
	\State $t\Leftarrow {\rm length}(L)$
 	\For{$(j_1,j_2,\dots,j_t)\in\{0,1,\dots,p-1\}^t$}
		\If{$\left|\beta_i-\sum_{k=1}^{t}{j_kL[k]}\right|_p<\left|\beta_1\right|_p$}
			\State $\beta_i^{\prime}\Leftarrow\beta_i-\sum_{k=1}^{t}{j_kL[k]}$
			\State back to step 3 with the next $i$
		\EndIf
	\State $\beta_i^{\prime}\Leftarrow\beta_i$
	\State append $\beta_i$ to $L$
	\EndFor
\EndFor
\State let $s$ be the subscript such that
$\left|\beta_s^{\prime}\right|_p=\max{\left\{\left|\beta_i^{\prime}\right|_p:\beta_i\not\in L\right\}}$
\State\Return $\lambda_2\Leftarrow\left|\beta_s^{\prime}\right|_p$, $v\Leftarrow\beta_s^{\prime}$
\end{algorithmic}
\end{algorithm}

This algorithm takes at most $(m-1)p^f$ many $p$-adic absolute value computations of elements of $K$. Compared with the LVP algorithm given in \cite{ref-1',ref-1}, which takes $O(p^m)$ many $p$-adic absolute value computations, our algorithm is much faster when $f$ is much less than $m$.

\section{Break the Schemes}

\subsection{The Attack}\label{se-4.1}

In \cite{ref-2}, Deng et al. discussed four kinds of possible attacks on their scheme: recovering a uniformizer, finding an orthogonal basis, solving CVP-instances and modulo $p$ attack. They conjectured that these attacks are hard to implement. However, as we have shown in the previous section, it is easy to find an orthogonal basis of some special lattices in a totally ramified field when $p$ is not very large.\par
Our attack is exactly based on recovering a uniformizer and finding an orthogonal basis. To break the schemes, we first recover a uniformizer from the public key and then use it to generate an orthogonal basis, as presented in the following lemma and theorem. 

\begin{lemma}\label{le-4.6}
Let $F(x)$ be the polynomial in the public key of the signature scheme or the encryption cryptosystem. Let $\zeta$ be a root of $F(x)$ and $\mathcal{L}=\mathcal{O}_K=\mathbb{Z}_p[\zeta]$. Then there is a deterministic algorithm to find the value $\lambda_2$ and the lattice vector $v\in\mathcal{L}$ such that $\left|v\right|_p=\lambda_2$. The algorithm takes $O(pm)$ many $p$-adic absolute value computations of elements of $K$.
\end{lemma}
\begin{proof}
Let $(\alpha_1,\alpha_2,\dots,\alpha_n)$ be the private key in the signature scheme or the encryption cryptosystem. Since $\alpha_i=\theta^{j_i}$ where $\theta$ is a uniformizer and $j_i \pmod n$ are distinct, we may assume that
$$\left|\alpha_1\right|_p>\left|\alpha_2\right|_p>\cdots>\left|\alpha_n\right|_p.$$
Also, we have
$$\left|\alpha_n\right|_p=\left|\theta^{j_n}\right|_p=p^{-\frac{j_n}{n}}>p^{-1}>\left|p\alpha_1\right|_p.$$
Since
$$\mathcal{L}=\mathcal{L}(\alpha_1,\alpha_2,\dots,\alpha_n)=\mathcal{L}(1,\zeta,\zeta^2,\dots,\zeta^{n-1}),$$
by Theorem \ref{th-4.2}, there is a deterministic algorithm to find the value $\lambda_2$ and the lattice vector $v\in\mathcal{L}$ such that $\left|v\right|_p=\lambda_2$. The algorithm takes $O(pm)$ many $p$-adic absolute value computations of elements of $K$.
\end{proof}

Uniformizers are just the second longest vector in the $p$-adic lattice $\mathcal{O}_K$. Once we find out a uniformizer of $K$, we can use this uniformizer to generate an orthogonal basis of $\mathcal{O}_K$ easily.

\begin{theorem}\label{th-4.7}
Let $F(x)$ be the polynomial in the public key of the signature scheme or the encryption cryptosystem. Let $\zeta$ be a root of $F(x)$ and $\mathcal{L}=\mathcal{O}_K=\mathbb{Z}_p[\zeta]$. Then there is a deterministic algorithm to find an orthogonal basis of $\mathcal{L}$. The algorithm takes $O(pm)$ many $p$-adic absolute value computations of elements of $K$.
\end{theorem}
\begin{proof}
By Lemma \ref{le-4.6}, there is a deterministic algorithm to find the value $\lambda_2$ and the lattice vector $v\in\mathcal{L}$ such that $\left|v\right|_p=\lambda_2$. The algorithm takes $O(pm)$ many $p$-adic absolute value computations of elements of $K$. By Proposition \ref{pr-2.3}, $1,v,v^2,\dots,v^{n-1}$ is an orthogonal basis of $\mathcal{L}$.
\end{proof}

Once we obtain an orthogonal basis of $\mathcal{O}_K$, though it is not identical to the private key in general, we can use this orthogonal basis to solve the LVP and the CVP in the lattice $\mathcal{O}_K$ efficiently. This enables us to forge a valid signature of any message and decrypt any ciphertext. Hence the public-key encryption cryptosystem and signature scheme are broken completely.\par
On the other hand, if we want to use the original LVP algorithm in \cite{ref-1',ref-1} to recover a uniformizer, then it will take $O(p^m)$ many $p$-adic absolute value computations of elements of $K$. Therefore, this attack extremely reduces the security strength of the schemes.

\subsection{A More Effective Way}\label{se-4.2}

When the degree of the polynomial in the public key is prime to $p$, the time consumption of finding a uniformizer can be reduced to $O(1)$. We have the following theorem.

\begin{theorem}\label{th-4.8}
Let
$$F(x)=x^n+F_{n-1}x^{n-1}+\cdots+F_1x+F_0$$
be the polynomial in the public key of the signature scheme or the encryption cryptosystem. Let $\zeta$ be a root of $F(x)$ and $K=\mathbb{Q}_p(\zeta)$. If $\gcd(n,p)=1$, then $\zeta+(F_{n-1}n^{-1}\mod p)$ is a uniformizer.
\end{theorem}
\begin{proof}
Let $\zeta=\zeta_1,\zeta_2,\dots,\zeta_n$ be $n$ roots of $F(x)$. Let $\theta$ be a uniformizer. Write
$$\zeta=a_{0}+a_{1}\theta+a_{2}\theta^2+\cdots+a_{n-1}\theta^{n-1}$$
where $a_{i}\in\mathbb{Z}_p$ $(0\le i\le n-1)$. Let $L$ be the splitting field of $F$ over $\mathbb{Q}_p$. Then for each $\zeta_i$ ($1\le i\le n$), there is a Galois automorphism $\sigma_i\in {\rm Gal}(L/\mathbb{Q}_p)$ which maps $\zeta$ to $\zeta_i$, i.e.,
$$\zeta_i=\sigma_i(\zeta)=a_{0}+a_{1}\sigma_i(\theta)+a_{2}\sigma_i(\theta)^2+\cdots+a_{n-1}\sigma_i(\theta)^{n-1}.$$
We have
$$-F_{n-1}=\sum^{n}_{i=1}{\zeta_i}=na_0+a_1\sum^{n}_{i=1}{\sigma_i(\theta)}+\cdots+a_{n-1}\sum^{n}_{i=1}{\sigma_i(\theta)^{n-1}}.$$

Let
$$f(x)=x^n+f_{n-1}x^{n-1}+\cdots+f_{1}x+f_0$$
be the minimal polynomial of $\theta$. Since $\sigma_i\ne\sigma_j$ for $i\ne j$, $\theta=\sigma_1(\theta),\sigma_2(\theta),\dots,\sigma_n(\theta)$ compose $n$ roots of $f(x)$. By fundamental theorem of symmetric polynomials, symmetric polynomials $\sum^{n}_{i=1}{\sigma_i(\theta)^{j}}$ ($1\le j\le n-1$)  can be given by a polynomial expression in terms of elementary symmetric polynomials, i.e., in terms of $f_0,f_1,\dots,f_{n-1}$. Since $\theta$ is a uniformizer, $f(x)$ is an Eisenstein polynomial and hence $f_i\equiv0\pmod p$ for $0\le i\le n-1$. Therefore,
$$a_{0}\equiv -F_{n-1}n^{-1}\pmod p.$$
Finally, since $a_{1}\not\equiv0\pmod p$, we conclude that $\zeta+(F_{n-1}n^{-1}\mod p)$ is a uniformizer.
\end{proof}

\subsection{A Toy Example}\label{se-4.3}

In order to illustrate our attack on the $p$-adic public-key encryption cryptosystem, we conduct an attack on the simple example given in \cite{ref-2}, Section 7.5.\par
Suppose that Oscar obtains the public key $(F(x),\delta=\frac{1}{5},(\beta_1,\beta_2,\beta_3,\beta_4))$, where
\begin{equation*}
\begin{split}
F(x)&=x^{20}-20x^{19}+190x^{18}-1120x^{17}+4555x^{16}-13470x^{15}+29670x^{14}\\
&-48500x^{13}+54972x^{12}-26650x^{11}-57366x^{10}+202684x^9-378052x^8\\
&+504970x^7-502444x^6+370306x^5-200173x^4+79034x^3-21942x^2\\
&+3548x-167
\end{split}
\end{equation*}
is the minimal polynomial of $\zeta$ and $\beta_1,\beta_2,\beta_3,\beta_4$ are expressed as polynomials in $\zeta$ of degree less than $n$ with coefficients in $\mathbb{Z}_2$. He runs Algorithm \ref{al-1} with the input $F(x)$. Since $\left|\zeta^i\right|_2=1$ for $i=0,1,\dots,19$ and 
\begin{equation*}
\left|\zeta^i-1\right|_2=\left\{
\begin{array}{lcl}
2^{-\frac{1}{20}} & \mbox{for} & i=1,3,5,7,9,11,13,15,17,19\\
2^{-\frac{1}{10}} & \mbox{for} & i=2,6,10,14,18\\
2^{-\frac{1}{5}} & \mbox{for} & i=4,12\\
2^{-\frac{2}{5}} & \mbox{for} & i=8\\
2^{-\frac{4}{5}} & \mbox{for} & i=16\\
\end{array}\right.,
\end{equation*}
the algorithm outputs $\lambda_2=2^{-\frac{1}{20}}$ and a vector $\zeta-1$, which is a uniformizer of $\mathbb{Q}_2(\zeta)$. Denote $\gamma=\zeta-1$, then $1,\gamma,\gamma^2,\dots,\gamma^{19}$ is an orthogonal basis of $\mathbb{Z}_2[\zeta]$.\par
The ciphertext is $C=\beta_1+\beta_2+\beta_4+r$, where
\begin{equation*}
\begin{split}
r&=\zeta^{19}+\zeta^{18}+3\zeta^{16}+2\zeta^{15}+3\zeta^{14}+\zeta^{13}+3\zeta^{12}\\
&+2\zeta^{10}+3\zeta^9+2\zeta^7+2\zeta^6+3\zeta^4+3\zeta^3+1.
\end{split}
\end{equation*}
$C$ is also expressed as polynomials in $\zeta$ and Oscar does not know the above summands. To recover the plaintext from $C$, Oscar first expresses $\beta_1,\beta_2,\beta_3,\beta_4$ and $C$ as polynomials in $\gamma$. The representations of $\beta_1,\beta_2,\beta_3,\beta_4$ and $C$ are displayed in Appendix \ref{se-B}.

Next, Oscar runs the algorithm ``CVP with orthogonal basis'' in \cite{ref-3.5} with the orthogonal basis $1,\gamma,\gamma^2,\dots,\gamma^{19}$, the lattice $\mathcal{L}(\beta_1,\beta_2,\beta_3,\beta_4)$ and the target vector $C$. The algorithm outputs a lattice vector $v=-\beta_1+\beta_2+\beta_4$ which is
closest to $C$. Therefore, the plaintext is
$$(-1,1,0,1)\equiv(1,1,0,1)\pmod {2}.$$
Oscar successfully recovers the plaintext from the ciphertext.\par
Of course, our attack can be used to break $p$-adic lattice encryption cryptosystems and signature schemes whose Eisenstein polynomial has larger degrees. However, the coefficients of the ciphertext and the basis in the public key are too large to display. Therefore, rather than presenting the details, we list the time consumption of finding a uniformizer in Table \ref{tab-1} instead.\par

\begin{table}[htbp]
  \centering  
    \begin{tabular}{|c|c|c|c|}
    \hline
    \diagbox{$n$}{$p$} & 5     & 7     & 11 \bigstrut\\
    \hline
    200   & 125   & 141   & 391 \bigstrut\\
    \hline
    300   & 564   & 658   & 1110 \bigstrut\\
    \hline
    500   & 2172  & 3234  & 5219 \bigstrut\\
    \hline
    1000  & 22110 & 28220 & 53703 \bigstrut\\
    \hline
    1500  & 72923 & 119344 & 212610 \bigstrut\\
    \hline
    \end{tabular}%
   \caption{Time consumption(ms) of finding a uniformizer for different $n$'s and $p$'s}
  \label{tab-1}%
\end{table}%

The Eisenstein polynomial $f(x)$ is of form $x^n+\sum_{i=0}^{n-1}px^i$. When $n<500$, $\zeta$ is chosen randomly. When $n\ge500$, it is too slow to compute the minimal polynomial for a random $\zeta$, hence we fix $\zeta$ to be $-1+x-x^3+x^{100}$.  This experiment is done on a personal
computer with Windows 11 operation system, i7-12700 CPU and 16-GB memory.

\section{Conclusion}\label{se-6}
Our attack extremely reduces the security strength of the $p$-adic lattice public-key encryption cryptosystems and signature schemes. It succeeds because the extension field used in the signature scheme and the public-key encryption cryptosystem is totally ramified. In order to avoid this attack, we suggest to add residue degree to the extension field. As shown in Section \ref{se-4.1.2}, our attack is not efficient when the extension field has a large residue degree. According to \cite{ref-1}, the LVP is also easy for an unramified extension, since we can choose $p\alpha_i$ as a second longest vector if the basis vector $\alpha_i$ is of length $\lambda_1$. Therefore, we recommend that the extension field should have large  ramification index and residue degree which are almost the same.\par
In a totally ramified extension field $K$, the uniformizer $\pi$ generates an orthogonal basis of $\mathcal{O}_K$. However, in a general extension field $K$, we can not find an orthogonal basis of $\mathcal{O}_K$ as easily as in a totally ramified extension field. Therefore, if we want to design such a scheme, constructing an orthogonal basis of $\mathcal{O}_K$ will be the main difficulty. It is worth for further study and there is much work to do.

\backmatter

\bmhead{Acknowledgements}

This work was supported by National Natural Science Foundation of China(No. 12271517) and National Key R\&D Program of China(No. 2020YFA0712300).

\section*{Statements and Declarations}

The author has no relevant financial or non-financial interests to disclose.

\begin{appendices}

\section{Proof for the efficiency}\label{se-A}

\begin{lemma}\label{le-A.1}
Write
$$\zeta=a_0+a_1\theta+a_2\theta^2+\cdots+a_{n-1}\theta^{n-1}$$
where $a_i\in\mathbb{Z}_p$ $(0\le i\le n-1)$. Then $\mathbb{Z}_p[\theta]=\mathbb{Z}_p[\zeta]$ if and only if $a_1\not\equiv0\pmod p$.
\end{lemma}
\begin{proof}
Suppose that $a_1\not\equiv0\pmod p$, then $\zeta-a_0$ is also a uniformizer. Hence
$$\mathbb{Z}_p[\zeta]=\mathbb{Z}_p[\zeta-a_0]=\mathbb{Z}_p[\theta].$$\par
On the other hand, suppose that $\mathbb{Z}_p[\theta]=\mathbb{Z}_p[\zeta]$, we can write
\begin{equation*}
\begin{split}
\theta&=b_0+b_1\zeta+b_2\zeta^2+\cdots+b_{n-1}\zeta^{n-1}\\
&=b_0+b_1(a_0+a_1\theta+a_2\theta^2+\cdots+a_{n-1}\theta^{n-1})+\cdots\\
&+b_{n-1}(a_0+a_1\theta+a_2\theta^2+\cdots+a_{n-1}\theta^{n-1})^{n-1}
\end{split}
\end{equation*}
where $b_i\in\mathbb{Z}_p$ $(0\le i\le n-1)$. If $a_1\equiv0\pmod p$, then we have
\begin{equation*}
\begin{split}
\theta&\equiv b_0+b_1(a_0+a_2\theta^2+\cdots+a_{n-1}\theta^{n-1})+\cdots\\
&+b_{n-1}(a_0+a_2\theta^2+\cdots+a_{n-1}\theta^{n-1})^{n-1}\pmod p.
\end{split}
\end{equation*}
Since the minimum polynomial of $\zeta$ over $\mathbb{Q}_p$ is an Eisenstein polynomial, the coefficients of term $\theta$ that appears in the right side of the above congruence expression are all multiples of $p$. This is a contradiction. Therefore, $a_1\not\equiv0\pmod p$.
\end{proof}

\begin{lemma}\label{le-A.2}
For any $t=H(M\|r)$, we have $\left|t-v\right|_p<1$.
\end{lemma}
\begin{proof}
Since $K=\mathbb{Q}_p(\theta)$, we can write
$$t=a_1\alpha_1+a_2\alpha_2+\cdots+a_n\alpha_n$$
where $a_i\in\mathbb{Q}_p$, $(1\le i\le n)$. Since $\alpha_1=1\in\mathcal{L}$, we have
$$\left|t-v\right|_p\le\left|t-\lfloor a_1\rfloor\alpha_1\right|_p<1.$$
\end{proof}
\section{Parameters in the toy example}\label{se-B}

Denote $c=755873885678037304696930874820307$, then $\beta_1=1$, $\beta_2=\gamma+1$,
\begin{equation*}
\begin{split}
\beta_3
& = -27108923760453500594114368147861/c\cdot\gamma^{19}\\
& + 6796059925906838569313127002344/c\cdot\gamma^{18}\\
& - 2894797737785704486333606848661/c\cdot\gamma^{17}\\
& - 541149271940293025229720323609264/c\cdot\gamma^{16}\\
& - 1219843153026858133885726090339770/c\cdot\gamma^{15}\\
& - 2808839733549013789844629261746068/c\cdot\gamma^{14}\\
& - 5313226953147016431704012154556784/c\cdot\gamma^{13}\\
& - 13157773442114268257340821253914084/c\cdot\gamma^{12}\\
& - 10548221619842211335147232490989965/c\cdot\gamma^{11}\\
& - 34381480844481901178070560983659504/c\cdot\gamma^{10}\\
& - 40457680366526053232506218942648521/c\cdot\gamma^9\\
& - 40324974837517376176270199515893136/c\cdot\gamma^8\\
& - 61919621949620553697389271288015350/c\cdot\gamma^7\\
& - 49577040839384885080983260133341290/c\cdot\gamma^6\\
& - 33420461054937491965850059498110469/c\cdot\gamma^5\\
& - 63334437339121387216907456988701600/c\cdot\gamma^4\\
& - 35485142914116339830882363036189729/c\cdot\gamma^3\\
& + 13319126479638277384948035705525862/c\cdot\gamma^2\\
& + 2307592141409163879671621342840038/c\cdot\gamma\\
& - 950075552624500643892558174049809/c,
\end{split}
\end{equation*}

\begin{equation*}
\begin{split}
\beta_4
& = -14485353108585552555664586328940/c\cdot\gamma^{19}\\
& + 6920651491155326119230629775288/c\cdot\gamma^{18}\\
& - 2768553635476272399349240377310/c\cdot\gamma^{17}\\
& - 288147212249787451567048006654254/c\cdot\gamma^{16}\\
& - 586697588755490712256433699640138/c\cdot\gamma^{15}\\
& - 1360136668305622477259338235646776/c\cdot\gamma^{14}\\
& - 2505102612122136411591027444054762/c\cdot\gamma^{13}\\
& - 6409095083612775651811538584361845/c\cdot\gamma^{12}\\
& - 4065989881207233334221617665320076/c\cdot\gamma^{11}\\
& - 17189563342011977016894330171381732/c\cdot\gamma^{10}\\
& - 17344915362034791698575633181051724/c\cdot\gamma^9\\
& - 16990260206579698534950129752547810/c\cdot\gamma^8\\
& - 27943965227519474547168258794651124/c\cdot\gamma^7\\
& - 18949419997068447949384315791092050/c\cdot\gamma^6\\
& - 11772018947343579357384864149315374/c\cdot\gamma^5\\
& - 29002780995574442228220988357294545/c\cdot\gamma^4\\
& - 10758271860130185122438319031482080/c\cdot\gamma^3\\
& + 12057724496469654972080098544477460/c\cdot\gamma^2\\
& + 445483857595665959192001502180342/c\cdot\gamma\\
& + 136786440750979030757408189729633/c,
\end{split}
\end{equation*}

\begin{equation*}
\begin{split}
C
& = 741388532569451752141266288491367/c\cdot\gamma^{19}\\
& + 15124398365051901420057848126181428/c\cdot\gamma^{18}\\
& + 142857395839513574315320586100660713/c\cdot\gamma^{17}\\
& + 850069974175542180332480186166191121/c\cdot\gamma^{16}\\
& + 3583767268296497408160589774698255656/c\cdot\gamma^{15}\\
& + 11397973933242174969657155184929403091/c\cdot\gamma^{14}\\
& + 28443296837109455751248008584668558569/c\cdot\gamma^{13}\\
& + 57145971274918465869787828837453870732/c\cdot\gamma^{12}\\
& + 94061415587207823186980038030222604846/c\cdot\gamma^{11}\\
& + 128220594511365406951340913237467802353/c\cdot\gamma^{10}\\
& + 145663740819137451017150305732686357000/c\cdot\gamma^{9}\\
& + 138094774383113040541078865455664306616/c\cdot\gamma^{8}\\
& + 108923717916404777624468448037804399856/c\cdot\gamma^{7}\\
& + 70968946553455805022662878792958039855/c\cdot\gamma^{6}\\
& + 37714649489129176335371132029007027303/c\cdot\gamma^{5}\\
& + 16018199811949157536487621484077823065/c\cdot\gamma^{4}\\
& + 5303035144456472066896985730955276130/c\cdot\gamma^{3}\\
& + 1327278285576254565144739820731811640/c\cdot\gamma^{2}\\
& + 227963523446684894672968194823092749/c\cdot\gamma\\
& + 24324750782448172781059196183979457/c.
\end{split}
\end{equation*}




\end{appendices}



\end{document}